\documentclass[letterpaper]{article} % DO NOT CHANGE THIS
\usepackage[preprint]{aaai2027}  % DO NOT CHANGE THIS
\usepackage[hyphens]{url}  % DO NOT CHANGE THIS
\usepackage{graphicx} % DO NOT CHANGE THIS
\usepackage{natbib}  % DO NOT CHANGE THIS AND DO NOT ADD ANY OPTIONS TO IT
\usepackage{caption} % DO NOT CHANGE THIS AND DO NOT ADD ANY OPTIONS TO IT
\usepackage[linesnumbered,ruled,lined,noend]{algorithm2e}

\SetKwComment{Comment}{// }{}

\usepackage{booktabs}

\usepackage{amsmath}
\usepackage{amssymb}
\usepackage{amsthm}
\usepackage{multirow}
\usepackage{subcaption}
\usepackage{tikz}
\usepackage[dvipsnames,table]{xcolor}
\usepackage{dsfont}
\usepackage{mathtools}

\usetikzlibrary{shapes,arrows}

\newtheorem{corollary}{Corollary}
\newtheorem{proposition}{Proposition}
\newtheorem{theorem}{Theorem}

\theoremstyle{remark}
\newtheorem{remark}{Remark}

\allowdisplaybreaks
\newcommand{\argmax}{\operatornamewithlimits{argmax}}

\title{Watermarked Game Solving via Perturbed Regret Minimization}
\author{
    Juho Kim\textsuperscript{\rm 1},
    Tuomas Sandholm\textsuperscript{\rm 1,2,3,4}
}
\affiliations{
    \textsuperscript{\rm 1}Computer Science Department, CMU\\
    \textsuperscript{\rm 2}Strategic Machine, Inc.\\
    \textsuperscript{\rm 3}Strategy Robot, Inc.\\
    \textsuperscript{\rm 4}Optimized Markets, Inc.\\
    \{juhok,sandholm\}@cs.cmu.edu
}

\begin{document}

\maketitle

\begin{abstract}
	Many real-world interactions among self-interested parties can be modeled by game theory, and the rapid advancements in AI have raised concerns about the possible misuse---accidental or deliberate---of superhuman or human-level game-playing agents by bad actors.
	While AI watermarking has mainly been applied to LLM-generated texts, a recent line of work proposes developing watermarking techniques for agents in game-theoretic settings.
	However, existing watermarking techniques for game-theoretic agents are not readily applicable due to their limited scope or capabilities---they are tailored to perfect-information games and are thus inapplicable to richer game types.
	We propose a new approach to watermarking game-playing agents, which a) can be applied to imperfect-information settings; b) is directly integrated into the learning process itself; and c) incurs only a bounded cost in exploitability.
	For this purpose, we introduce perturbed regret minimization, which adds perturbations to the utilities prior to observation so as to encourage the learning algorithm to embed the watermark.
	Our experiments show that the watermark incurs only a small exploitability cost and can be detected within just a couple of hours of gameplay at human speed.
\end{abstract}


\section{Introduction}

Recent advances in AI have given rise to remarkably capable models that radically transform the world. 
However, they also have the potential to disrupt society by empowering bad actors to automate cyberattacks~\cite{anthropic}, execute large-scale phishing scams~\cite{schulz}, spread mis/disinformation~\cite{qiaoetal,lancet}, make biological weapons~\cite{wellsandramkumar}, \textit{etc}.
To at least partially address these risks, the AI research community has paid significant attention to watermarking large language models (LLMs), which `mark' the output of an LLM with a hidden, robust signature to enable verification of whether or not a piece of text was generated using the LLM~\cite{dathathrietal,kirchenbaueretal,zhao,takezawaetal,huoetal,wangetal}.
However, the potential societal harms of AI are not limited to those involving LLMs.

Indeed, considering how many real-life interactions between self-interested parties in domains like business, finance, political science, economics, national defense, \textit{etc.} can be modeled by game theory, many important areas of AI misuse involve game-playing agents.
One possible application of a watermark for game-playing agents is \textit{detecting unauthorized use of AI tools} in online platforms, for example, those for recreational games like chess and poker that already face a widespread cheating problem in terms of botting~\cite{chappell,grant}.
Another more high-stakes scenario of its potential use includes \textit{preventing or detecting intellectual property (IP) theft} (\textit{e.g.}, license noncompliance~\cite{stockfish}, cybersecurity breaches, or distillation learning): the mere presense of a watermark can deter model theft; or, by collecting information about a competitor's model in the wild, one can ascertain whether it one's own model or at least derived from it or not~\cite{sander}.
Additionally, a researcher can use a watermark detector to \textit{filter out non-human data} from a dataset of `gameplays' so it can be used to study human behavior.

Addressing this gap,~\citet{kimetal} recently connected the problem of text generation with game playing to extend LLM watermarking techniques, namely the KGW watermark~\cite{kirchenbaueretal}, for watermarking game-playing agents.
However, their contribution works for perfect-information games only and thus fails to adequately address the myriad of settings that occur in real-life, which include the presence of private information only available to a subset of the parties involved.
Their technique cannot be applied in such imperfect-information games.

In this paper, we explore the possibility of simultaneous watermarking and (imperfect-information) game solving, thus embedding the watermark during the learning process.
In our methodology, we focus on the framework of regret minimization, which is the most common family of techniques for solving games in a manner that is both theoretically sound and amenable to large-scale settings: this family of algorithms was instrumental to creating near-optimal or superhuman agents for playing milestone imperfect-information games like poker~\cite{bowlingetal,moravciketal,brownandsandholm2018,brownandsandholm2019}, The Resistance: Avalon~\cite{serrinoetal}, and dark chess~\cite{zhangandsandholm}. 
While our watermarking technique shares some core ideas with the foundational KGW watermark for LLMs~\cite{kirchenbaueretal} and its recent adaptation for perfect-information games~\cite{kimetal}, we introduce several important differences from the prior works.
Most notably, the watermark is learned through what we call \textit{perturbed regret minimization}, which applies a small perturbation to utilities prior to observation, thereby encouraging certain actions to be played more frequently than others.
In our analysis, we provide a bound for the exploitability incurred by watermarking, and show the fundamental tradeoff between the exploitability cost and watermark detectability.
In our experiments, we apply our technique to selections of large games and show that a) the incurred exploitability cost is small and b) the watermark can be detected already within a couple of hours of gameplay at human speed.

\section{Notation and Background}

In this section, we define the notation used throughout this paper and review game solving and watermarking.

\subsection{Extensive-Form Games (EFGs)}

An \textit{extensive-form game (EFG)} contains a finite set of histories $H$ and players $P$.
Chance $p_c$ is a special type of player not in $P$ that represents random events.
Every $h \in H$ is a sequence of actions taken by each player, and we denote the player in turn to act on $h$ as $p(h)$ and the set of available actions on $h$ as $A(h)$.
Further, we use $h \cdot a = h'$ to denote that history $h'$ is reached by applying $a \in A(h)$ at $h$.
If $p(h) = p_c$, that is, $h$ is a chance node, then $f_c(h, a)$ gives a fixed probability distribution over $A(h)$.
Every terminal history $z \in Z \subseteq H$ is associated with a utility $u_i(z)$ for each non-chance player $i \in P$.

In EFGs, private information belonging to a subset of players is represented by partitioning non-chance histories into \textit{information sets (infosets)} $\mathcal{I}$.
For each $i \in P$, let $\mathcal{I}_i$ be a subset of $\mathcal{I}$ belonging to $i$.
Here, $i$ cannot distinguish between any $h, h' \in I \in \mathcal{I}_i$, so $A(h) = A(h')$, and we denote the common set by $A(I)$.
Each $i \in P$ plays with strategy $\sigma_i$, which, for any $I \in \mathcal{I}_i$, gives a probability distribution over $A(I)$.
$\sigma = (\sigma_1, \ldots, \sigma_n)$ is a strategy profile.

In \textit{two-player zero-sum (2p0s) games}, there are two players and the sum of player utilities at any terminal node is zero.

\subsection{Nash Equilibrium (NE)}

A traditional solution concept in game theory is \textit{Nash equilibrium (NE)}, where no player can gain by deviating, \textit{i.e.}, an NE is a strategy profile $\sigma^*$ such that
\[
	\forall i \in P, \sigma_i \in S_i: u_i(\sigma^*) \ge u_i(\sigma_i, \sigma_{-i}^*),
\]
where $S_i$ is the strategy space of player $i$.
An NE is guaranteed to exist~\cite{nash}.
Its approximation is \textit{$\epsilon$-Nash equilibrium ($\epsilon$-NE)}, where the amount any player can gain by deviating is bounded, \textit{i.e.}, an $\epsilon$-NE is a strategy profile $\sigma^*$ such that
\[
	\forall i \in P, \sigma_i \in S_i: u_i(\sigma^*) + \epsilon \ge u_i(\sigma_i, \sigma_{-i}^*).
\]

In 2p0s games, a common way to measure the quality of a strategy profile is \textit{exploitability}, which, intuitively, is the strategic distance between a given strategy profile and an NE. 
The exploitability of a strategy profile $\sigma^*$ is defined as follows:
\[
	\epsilon = \frac1{|P|} \sum_{i \in P} \max_{\sigma_i \in S_i} u_i(\sigma_i, \sigma_{-i}^*) - u_i(\sigma^*).
\]
Note that a strategy profile with exploitability $\epsilon$ is a $2\epsilon$-NE.

\subsection{Tree-Form Sequential Decision Process (TFSDP)}

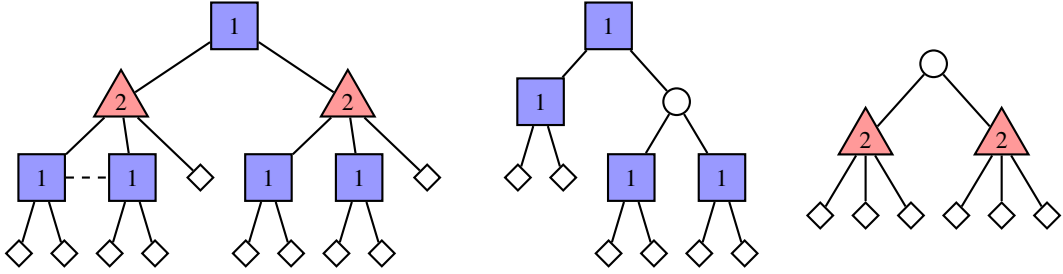
\begin{figure*}[t!]
	\centering
	\begin{tikzpicture}
		\begin{scope}
			\node[rectangle, draw, minimum size=1.75em, thick, inner sep=0, fill=blue!40] (h1) at (2.85, 0) {\small 1};

			\node[isosceles triangle, isosceles triangle apex angle=60, draw, shape border rotate=90, minimum size=1.75em, thick, inner sep=0, fill=red!40] (h2) at (1.35, -1) {\small 2};
			\node[isosceles triangle, isosceles triangle apex angle=60, draw, shape border rotate=90, minimum size=1.75em, thick, inner sep=0, fill=red!40] (h3) at (4.35, -1) {\small 2};

			\node[rectangle, draw, minimum size=1.75em, thick, inner sep=0, fill=blue!40] (h4) at (0.3, -2) {\small 1};
			\node[rectangle, draw, minimum size=1.75em, thick, inner sep=0, fill=blue!40] (h5) at (1.5, -2) {\small 1};
			\node[diamond, draw, minimum size=1em, thick, inner sep=0] (z1) at (2.4, -2) {};
			\node[rectangle, draw, minimum size=1.75em, thick, inner sep=0, fill=blue!40] (h6) at (3.3, -2) {\small 1};
			\node[rectangle, draw, minimum size=1.75em, thick, inner sep=0, fill=blue!40] (h7) at (4.5, -2) {\small 1};
			\node[diamond, draw, minimum size=1em, thick, inner sep=0] (z2) at (5.4, -2) {};

			\node[diamond, draw, minimum size=1em, thick, inner sep=0] (z3) at (0, -3) {};
			\node[diamond, draw, minimum size=1em, thick, inner sep=0] (z4) at (0.6, -3) {};
			\node[diamond, draw, minimum size=1em, thick, inner sep=0] (z5) at (1.2, -3) {};
			\node[diamond, draw, minimum size=1em, thick, inner sep=0] (z6) at (1.8, -3) {};
			\node[diamond, draw, minimum size=1em, thick, inner sep=0] (z7) at (3, -3) {};
			\node[diamond, draw, minimum size=1em, thick, inner sep=0] (z8) at (3.6, -3) {};
			\node[diamond, draw, minimum size=1em, thick, inner sep=0] (z9) at (4.2, -3) {};
			\node[diamond, draw, minimum size=1em, thick, inner sep=0] (z10) at (4.8, -3) {};

			\draw[thick] (h1) -- (h2);
			\draw[thick] (h1) -- (h3);

			\draw[thick] (h2) -- (h4);
			\draw[thick] (h2) -- (h5);
			\draw[thick] (h2) -- (z1);
			\draw[thick] (h3) -- (h6);
			\draw[thick] (h3) -- (h7);
			\draw[thick] (h3) -- (z2);

			\draw[thick] (h4) -- (z3);
			\draw[thick] (h4) -- (z4);
			\draw[thick] (h5) -- (z5);
			\draw[thick] (h5) -- (z6);
			\draw[thick] (h6) -- (z7);
			\draw[thick] (h6) -- (z8);
			\draw[thick] (h7) -- (z9);
			\draw[thick] (h7) -- (z10);

			\draw[dashed, thick] (h4) -- (h5);
		\end{scope}

		\begin{scope}[shift={(6.6, 0)}]
			\node[rectangle, draw, minimum size=1.75em, thick, inner sep=0, fill=blue!40] (j1) at (1.2, 0) {\small 1};

			\node[rectangle, draw, minimum size=1.75em, thick, inner sep=0, fill=blue!40] (j2) at (0.3, -1) {\small 1};
			\node[circle, draw, minimum size=1em, thick, inner sep=0] (k1) at (2.1, -1) {};

			\node[diamond, draw, minimum size=1em, thick, inner sep=0] (z1) at (0, -2) {};
			\node[diamond, draw, minimum size=1em, thick, inner sep=0] (z2) at (0.6, -2) {};
			\node[rectangle, draw, minimum size=1.75em, thick, inner sep=0, fill=blue!40] (j3) at (1.5, -2) {\small 1};
			\node[rectangle, draw, minimum size=1.75em, thick, inner sep=0, fill=blue!40] (j4) at (2.7, -2) {\small 1};

			\node[diamond, draw, minimum size=1em, thick, inner sep=0] (z3) at (1.2, -3) {};
			\node[diamond, draw, minimum size=1em, thick, inner sep=0] (z4) at (1.8, -3) {};
			\node[diamond, draw, minimum size=1em, thick, inner sep=0] (z5) at (2.4, -3) {};
			\node[diamond, draw, minimum size=1em, thick, inner sep=0] (z6) at (3, -3) {};

			\draw[thick] (j1) -- (j2);
			\draw[thick] (j1) -- (k1);

			\draw[thick] (j2) -- (z1);
			\draw[thick] (j2) -- (z2);
			\draw[thick] (k1) -- (j3);
			\draw[thick] (k1) -- (j4);

			\draw[thick] (j3) -- (z3);
			\draw[thick] (j3) -- (z4);
			\draw[thick] (j4) -- (z5);
			\draw[thick] (j4) -- (z6);
		\end{scope}

		\begin{scope}[shift={(10.6,0)}]
			\node[circle, draw, minimum size=1em, thick, inner sep=0] (k1) at (1.5, -0.5) {};

			\node[isosceles triangle, isosceles triangle apex angle=60, draw, shape border rotate=90, minimum size=1.75em, thick, inner sep=0, fill=red!40] (j1) at (0.6, -1.5) {\small 2};
			\node[isosceles triangle, isosceles triangle apex angle=60, draw, shape border rotate=90, minimum size=1.75em, thick, inner sep=0, fill=red!40] (j2) at (2.4, -1.5) {\small 2};

			\node[diamond, draw, minimum size=1em, thick, inner sep=0] (z1) at (0, -2.5) {};
			\node[diamond, draw, minimum size=1em, thick, inner sep=0] (z2) at (0.6, -2.5) {};
			\node[diamond, draw, minimum size=1em, thick, inner sep=0] (z3) at (1.2, -2.5) {};
			\node[diamond, draw, minimum size=1em, thick, inner sep=0] (z4) at (1.8, -2.5) {};
			\node[diamond, draw, minimum size=1em, thick, inner sep=0] (z5) at (2.4, -2.5) {};
			\node[diamond, draw, minimum size=1em, thick, inner sep=0] (z6) at (3.0, -2.5) {};

			\draw[thick] (k1) -- (j1);
			\draw[thick] (k1) -- (j2);

			\draw[thick] (j1) -- (z1);
			\draw[thick] (j1) -- (z2);
			\draw[thick] (j1) -- (z3);
			\draw[thick] (j2) -- (z4);
			\draw[thick] (j2) -- (z5);
			\draw[thick] (j2) -- (z6);
		\end{scope}
	\end{tikzpicture}
	\caption{
		The figure on the left represents an example 2p0s EFG.
		Square, triangle, and diamond nodes represent Player 1, Player 2, and terminal nodes, respectively.
		Edges represent actions, while the dashed line connects nodes in the same infoset.
		The middle and right figures represent the TFSDPs for Players 1 and 2, respectively.
		Square, circle, and diamond nodes represent decision points, observation points, and ends of the decision process, respectively.
		Edges represent actions or signals.
	}
	\label{fig:efg-tfsdp}
\end{figure*}

In EFGs, a player $i \in P$ faces a \textit{tree-form sequential decision process (TFSDP)}, as demonstrated in Figure~\ref{fig:efg-tfsdp}.
In a TFSDP, each node $p \in \mathcal P$ either is an end of the decision process $\bot$ or is in the set of decision points $\mathcal J$ or the set of observation points $\mathcal K$.
At a decision point $j \in \mathcal J$, an available action $a \in \mathcal A_j$ can be applied, and $\rho(j, a)$ denotes the child node reached from $j$ by applying $a$.
Similarly, a signal $s \in \mathcal{S}_k$ can be observed at an observation point $k \in \mathcal K$, and the corresponding child is $\rho(k, s)$.
Let $\Sigma^+ = \left\{(j, a): \forall j \in \mathcal J, a \in \mathcal A_j\right\}$ be the set of non-empty sequences, and let $\Sigma = \Sigma^+ \cup \{\emptyset\}$ be the set of (possibly empty) sequences, where $\emptyset$ is the empty sequence.
Note that decision points and infosets are equivalent.

Each time a TFSDP ends, the player observes utilities $\vec u \in \mathbb R^\Sigma$.
The utility associated with an empty sequence $\emptyset$ represents the utility realized when the TFSDP ends before any action is taken.
The utility associated with a non-empty sequence $(j, a) \in \Sigma^+$ represents the utility realized when the TFSDP ends immediately after taking $a$ at $j$.
Note that each element in $\vec u$ is already weighted by the probability of realization due to factors outside the TFSDP (\textit{e.g.}, strategies of other players and chance probabilities).

For each decision point $j$, let $p_j$ be its parent sequence (if $j$ does not have an ancestor decision point, then $p_j = \emptyset$).
For each $i \in P$, the strategy space $S_i$ is a sequence-form polytope~\cite{vonstengel,kolleretal}, defined as follows:
\[
	\resizebox{\columnwidth}{!}{$
		\mathcal X_i = \left\{\vec x \in \mathbb R_{\ge0}^\Sigma: \vec x[\emptyset] = 1, \forall j \in \mathcal J: \sum_{a \in \mathcal A_j} \vec x[(j, a)] = \vec x[p_j]\right\}.
	$}
\]
Thus, a learning algorithm that solves a TFSDP is said to operate over its sequence-form polytope.
One distinct advantage of representing a strategy in sequence-form is that, given a strategy $\vec x \in \mathcal X_i$ and a utility vector $\vec u \in \mathbb R^\Sigma$, the expected utility is simply $\langle\vec u, \vec x\rangle$, that is, a linear function.

\subsection{Regret Minimization}

The framework of \textit{regret minimization} has been integral to developing theoretically sound methods for solving games.
At each time $t \in \mathbb N^+$, a regret minimization algorithm outputs a strategy $\vec x^{(t)} \in \mathcal X$, where the strategy space $\mathcal X \subseteq \mathbb R^n$ is a compact and convex subset of the Euclidean space with some dimension $n$.
Then, the learning algorithm observes a linear utility vector $\vec u^{(t)}$ and receives utility $\langle\vec u^{(t)}, \vec x^{(t)}\rangle$.
The learning algorithm's performance is evaluated by its cumulative regret up to time $T \in \mathbb N$, defined as follows:
\[
	R_{\mathcal X}^{(T)} = \max_{\vec x \in \mathcal X} \sum_{t = 1}^T \langle\vec u^{(t)}, \vec x\rangle - \sum_{t = 1}^T \langle\vec u^{(t)}, \vec x^{(t)}\rangle.
\]
A regret minimizer is a learning algorithm that satisfies Hannan consistency: its average regret up to time $T$ tends to zero when taking the limit $T \to \infty$, that is, $R_{\mathcal X}^{(T)} = o(T)$.

One of the most common types of strategy space over which regret minimization is studied is a probability simplex.
We denote a probability simplex of dimension $m$ by $\Delta^m$.
However, a probability simplex is not rich enough to represent the strategy space of a player in EFGs, which is a sequence-form polytope.
\citet{farinaetal2019} demonstrated how a regret minimizer operating over a sequence-form polytope can be constructed using those operating over probability simplices by applying Cartesian product and convex hull operations; every regret minimizer operating over a probability simplex is placed at some decision point $j$ to mix the available actions $\mathcal A_j$.

When solving an EFG by minimizing regret on its TFSDPs, the average strategy profile converges to a \textit{coarse-correlated equilibrium}, which coincides with NE in 2p0s games. There are many variants of such algorithms. 
When \textit{regret matching (RM)}~\cite{hartandmascolell2000} is used locally, the construction reduces to \textit{counterfactual regret minimization (CFR)}~\cite{zinkevichetal2007}.
When RM\textsuperscript{+} is used the construction reduces to \textit{CFR\textsuperscript{+}}~\cite{tammelin2014}.
Finally, when discounted RM is used, the construction reduces to \textit{Discounted (DCFR)}~\cite{brownandsandholm2019aaai}.
The type of local regret minimizer is not important in this paper.

\subsection{Watermarking Game-Playing Agents}

Watermarking in the modern AI literature has focused on embedding a hidden, robust signature within LLM-generated texts.
In this section, we describe its recent adaptation for perfect-information game-playing agents by~\citet{kimetal}, henceforth called the KFS watermark.

The KFS watermark is designed to act as a wrapper to a given strategy profile $\sigma$ and accepts the following parameters: \textcolor{Green}{\textbf{green}}-list size $\lambda \in (0, 1)$ and hardness $\delta \ge 0$.
Each time a query is made for history $h$ (\textit{e.g.}, for action recommendations), the KFS watermark queries the underlying model of $\sigma$ to calculate (or estimate) the expected utilities over $A(h)$.
Next, the hash of the observation made at $h$ is used to seed a pseudo-random number generator, which is then used to randomly partition $A(h)$ into \textcolor{Red}{\textbf{red}} and \textcolor{Green}{\textbf{green}} lists (their sizes are informed by $\lambda$).
Then, the expected utilities of taking different actions are adjusted to favor the selection of actions in the \textcolor{Green}{\textbf{green}} list, as opposed to those in the \textcolor{Red}{\textbf{red}} list (the adjustments are informed by $\delta$).
Finally, the watermark returns the action associated with the maximum (estimated) expected utility.

The design of boosting the expected utilities associated with \textcolor{Green}{\textbf{green}} actions was inspired by how the KGW watermark~\cite{kirchenbaueretal} boosts the logits associated with \textcolor{Green}{\textbf{green}} words.
Also, the design of using observations to seed a pseudo-random number generator was inspired by how the KGW watermark uses prompts' final tokens to do so.
\citet{kimetal} analyzed the concentration and deterministic bounds on the loss in expected utility incurred by the watermark and demonstrated that the watermark does not meaningfully impact the strength of game-playing agents by experimenting with off-the-shelf chess engines.

\subsection{Shortcomings in the Literature}

There are two significant shortcomings in the literature on watermarking game-playing agents.
First, the KFS watermark is designed for perfect-information games.
However, most real-world interactions are of imperfect information, and it is not obvious whether the KFS approach can be adapted to this setting.
In perfect-information games, there is a clear way to obtain optimal action probabilities from boosted expected utilities (taking $\argmax$), but in imperfect-information games, translating modified expected utilities into optimal action probabilities is highly non-trivial. 
While it may be possible to come up with an ad hoc solution like modifying the action probabilities directly, it would be preferable to come up with a well-motivated solution with clear connections to a learning framework, as we do in this paper.

Second, the KFS analysis assumes that the opponent strategy stays fixed as we watermark our strategy. In other words, they assume that the opponent does not change its strategy to exploit the weaknesses in our strategy caused by our watermarking. We fix this issue by analyzing the quality of our watermarked strategy by evaluating how exploitable it is.

\section{Simultaneous Watermarking \\ and Game Solving}

We are now ready to present our contribution on \textit{simultaneous watermarking and game solving}.
The core goal of our watermark is similar to the KFS watermark for perfect-information games: we desire the gameplay of watermarked agents to contain a robust, hidden signature that can be detected.
The KFS watermark labels each action as either \textcolor{Red}{\textbf{red}} or \textcolor{Green}{\textbf{green}} and encourages the strategy profile to play more \textcolor{Green}{\textbf{green}} actions than would be expected by chance.
However, unlike the KFS watermark, which boosts the expected utilities of applying different actions, our watermark is learned through what we coin \textit{perturbed regret minimization}, where tiny label-dependent perturbations are applied to utilities prior to observation, thereby encouraging \textcolor{Green}{\textbf{green}} actions to be played with increased frequency.

Our perturbed regret minimization framework accepts two parameters: \textcolor{Green}{\textbf{green}}-list size $\gamma \in (0, 1)$ and hardness $\delta \ge 0$.
During initialization, the watermarked learning framework partitions the set of sequences $\Sigma$ into a \textcolor{Red}{\textbf{red}} list $R$ of size $(1 - \gamma)|\Sigma|$ and a \textcolor{Green}{\textbf{green}} list $G$ of size $\gamma|\Sigma|$.
When $\Sigma$ is too large even to be enumerated, sequences can be classified on-the-fly using a hash function (this is asymptotically equivalent).
Partitioning the set of sequences is similar in functionality to partitioning the set of actions at each decision point, except now a decision point can have more than $\gamma$ fraction of \textcolor{Green}{\textbf{green}} actions.
Our framework also modifies the regret minimization algorithm used to solve a given EFG, \textit{e.g.}, CFR.
Recall that a regret minimizer operating over a sequence-form polytope (\textit{i.e.}, a player's strategy space) is really a composition of local regret minimizers, each of which operates over a probability simplex and mixes the available actions at a decision point.
In our framework, we replace every local regret minimizer with its perturbed counterpart.

Let $\mathfrak R_j$ be a local regret minimizer at some decision point $j$.
The perturbed counterpart of $\mathfrak R_j$ is almost identical to the original, except that, at every time $t$, it observes a utility vector $\vec u_j^{(t)}$ with some perturbation $\vec w_j$, \textit{i.e.}, $\vec u_j^{(t)} + \vec w_j$, where
\[
	\vec w_j = \left(\begin{cases} \delta & (j, a) \in G \\ -\delta\gamma/(1 - \gamma) & (j, a) \in R \end{cases}\right)_{a \in A(j)}.
\]
Here, the perturbations are controlled by $\delta$ and $\gamma$ and are designed so as to encourage the learning algorithm to play actions of \textcolor{Green}{\textbf{green}} sequences more often than by chance.
We also added offsets for \textcolor{Red}{\textbf{red}} sequences to ensure that the total adjustment across all sequences is zero.

\subsection{Regret Minimizer under Perturbation (P/RM)}

In this subsection, we formally define a \textit{regret minimizer under perturbation (P/RM)}, which simply wraps a regret minimizer operating over a probability simplex and hijacks the utility to add certain parameterized noise.
Its regret circuit diagram is given in Figure~\ref{fig:perturbed-regret-minimizer}, and its pseudocode in Algorithm~\ref{alg:perturbed-regret-minimizer}.
The P/RM accepts two parameters: a component regret minimizer $\mathfrak R$ and a static perturbation $\vec w$.
In our context, the component regret minimizer is the regret minimizer in CFR responsible for mixing over actions at some decision point.
At time $t \in \mathbb N^+$, the P/RM forwards the output of the component regret minimizer as its own.
Then, the P/RM observes a utility vector $\vec u^{(t)}$ and adds a perturbation $\vec w$ to it.
Finally, the perturbed utility $\vec u^{(t)} + \vec w$ is forwarded to the component regret minimizer for its observation.

\begin{figure}[!h]
	\centering
	\begin{tikzpicture}
		\node[draw, dashed, minimum width=17em, minimum height=5em, very thick] (boundary) at (0em, 0em) {};
		\node (utility) at (-11em, 0) {$\vec u^{(t - 1)}$};
		\node (delta) at (-7em, 5em) {$\vec w$};
		\node[diamond, draw, minimum size=1em, thick, inner sep=0em] (gate) at (-7em, 0em) {};
		\node[draw, minimum width=4em, minimum height=2em, thick] (regret minimizer) at (4em, 0em) {$\mathfrak R$};
		\node (decision) at (11em, 0em) {$\vec x^{(t)}$};

		\draw (delta) edge[->, color=Green, ultra thick] (gate);
		\begin{scope}
			\clip (-7em, 5em) -- (-7em, 0em) -- (-5em, 5em) -- (-5em, 0em) -- cycle;
			\draw (delta) edge[->, color=Red, ultra thick] (gate);
		\end{scope}
		\draw (utility) edge[->, color=orange, very thick] (gate);
		\draw (gate) edge[->, color=orange, very thick] node[above, color=black] {$\vec u^{(t - 1)} + \vec w$} (regret minimizer);
		\draw (regret minimizer) edge[color=blue, very thick] (7.5em, 0em);
		\draw (regret minimizer) edge[->, color=blue, very thick] (decision);
	\end{tikzpicture}
	\caption{
		Regret circuit diagram of a P/RM with $\mathfrak R$ as its component and static perturbation $\vec w$.
	}
	\label{fig:perturbed-regret-minimizer}
\end{figure}
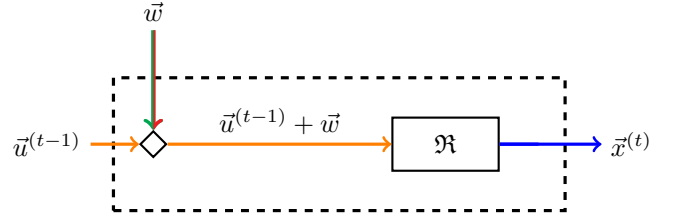

\begin{algorithm}[!h]
	\caption{Regret minimizer under perturbation (P/RM)}
	\label{alg:perturbed-regret-minimizer}
	\textbf{Input:} A component regret minimizer $\mathfrak R$ operating over $\Delta^m$ and a static perturbation $\vec w$.
	\BlankLine
	\textsc{NextStrategy}$()$: \\
	\Indp{
		\Return $\mathfrak R$.\textsc{NextStrategy}$()$\;
	}
	\Indm
	\BlankLine
	\textsc{ObserveUtility}$\left(\vec u^{(t)} \in \mathbb R^m\right)$: \\
	\Indp{
		$\mathfrak R$.\textsc{ObserveUtility}$(\vec u^{(t)} + \vec w)$\;
	}
	\Indm
\end{algorithm}

\subsection{Analysis}

We analyze P/RM in this section.
Due to limited space, proofs are relegated to Appendix~\ref{sec:omitted-proofs}.
In the following proposition, we show that a P/RM incurs only constant average regret.

\begin{proposition}
	Let $\mathfrak W$ be a P/RM with a component $\mathfrak R$ operating over $\Delta^m$ and a static perturbation in $[L, U]^m$.
	Further, let $W^{(T)}$ and $R^{(T)}$ be the regret incurred by $\mathfrak W$ and $\mathfrak R$, respectively, until time $T$.
	Then,
	\[
		W^{(T)} \le R^{(T)} + T(U - L).
	\]
	\label{ppn:probability-simplex}
\end{proposition}
\begin{remark}
	Assuming $\mathfrak R$ is a regret minimizer, $R^{(T)} = o(T)$, so the first term is not important.
	Technically, a P/RM is not a regret minimizer, as it incurs linear regret.
	However, the average regret can be bounded as small as desired by setting $L$ and $U$ accordingly.
	In our context, $L = -\delta\gamma/(1 - \gamma)$ and $U = \delta$.
	Thus, the average regret can be managed by setting $\delta$ accordingly, but it may impact the number of playthroughs needed to detect the watermark, as we will show later.
\end{remark}

Now, we use Proposition~\ref{ppn:probability-simplex} to analyze the regret incurred by perturbed regret minimization when operating on a TFSDP.
\begin{proposition}
	Let $\mathfrak W$ be the construction in perturbed regret minimization operating on a TFSDP with decision points $\mathcal J$.
	In every $j \in \mathcal J$, actions at $j$ are mixed by a P/RM with a component $\mathfrak R_j$ operating over $\Delta^{\mathcal A_j}$ and a static perturbation in $[L, U]^{\mathcal A_j}$.
	Further, let $W^{(T)}$ and $R_j^{(T)}$ be the regret incurred by $\mathfrak W$ and $\mathfrak R_j$, respectively, until time $T$.
	Then,
	\[
		W^{(T)} \le \sum_{j \in \mathcal J} \max\{0, R_j^{(T)}\} + T|\mathcal J|(U - L).
	\]
	\label{ppn:sequence-form-polytope}
\end{proposition}
\begin{remark}
	Again, in our context, the average regret can be managed by setting $\delta$ accordingly.
\end{remark}

Finally, we connect the propositions with the following theorem about the quality of the solution obtained via perturbed regret minimization.
\begin{theorem}
	Let $\Gamma$ be a 2p0s EFG with infosets $\mathcal I$.
	Suppose that perturbed regret minimization with \textcolor{Green}{\textbf{green}}-list size $\gamma \in (0, 1)$ and hardness $\delta \ge 0$ is run for $T$ iterations to solve $\Gamma$.
	Let $R_I^{(T)}$ be the regret incurred by the component of the P/RM at the decision point corresponding to infoset $I$ until time $T$.
	Then, the average strategy profile is a $2\epsilon$-NE where
	\[
		2\epsilon \le \sum_{I \in \mathcal I} \max\{0, R_I^{(T)}\}/T + \delta|\mathcal I|/(1 - \gamma).
	\]
	\label{thm:eps-ne}
\end{theorem}

Suppose that we fix the type of local regret minimizers to be regret matching.
Then, our framework reduces to a `watermarked' counterpart of CFR.

\begin{corollary}
	Let $\Gamma$ be a 2p0s EFG with infosets $\mathcal I$.
	Further, let $\Delta$ be the range of payoffs, and let $m$ be the maximum number of actions in a player node.
	Suppose that a watermarked counterpart of CFR with \textcolor{Green}{\textbf{green}}-list size $\gamma \in (0, 1)$ and hardness $\delta \ge 0$ is run for $T$ iterations to solve $\Gamma$.
	Then, the average strategy profile is a $2\epsilon$-NE where
	\[
		2\epsilon \le |\mathcal I|(\Delta + \delta/(1 - \gamma))\sqrt m/\sqrt T + \delta|\mathcal I|/(1 - \gamma).
	\]
	\label{cll:cfr}
\end{corollary}
\begin{remark}
	While it can be concerning that watermarking incurs some increase in exploitability, we later demonstrate that this is not much of a concern in practice.
	Besides, typical runs of the CFR family of algorithms (\textit{e.g.}, 1,000 iterations) yield solutions with some exploitability anyway.
	By managing the watermark parameters, the exploitability cost can be made unsubstantial in this context.
\end{remark}

\subsection{Detecting the Watermark}

Our watermark can be detected by calculating a z-score like its predecessors~\cite{kirchenbaueretal,zhao,takezawaetal,huoetal,wangetal,kimetal}.
As usual, the null hypothesis is that the player in question acts without any knowledge of the \textcolor{Green}{\textbf{green}} and \textcolor{Red}{\textbf{red}} lists.
Let $n$ be the number of actions applied by the player, and let $n_G$ be the number of \textcolor{Green}{\textbf{green}} actions applied by the player.
Then, we have the following z-score:
\[
	z = \frac{n_G - \gamma n}{\sqrt{n \gamma (1 - \gamma)}},
\]
which is a measure of how many more \textcolor{Green}{\textbf{green}} actions are applied than expected by chance.
\citet{kimetal}~suggested a threshold of $z = 4$, corresponding to a false positive rate of approximately $3 \times 10^{-5}$.

\section{Experiments}

\begin{figure*}[t]
	\centering
	\includegraphics[width=\textwidth]{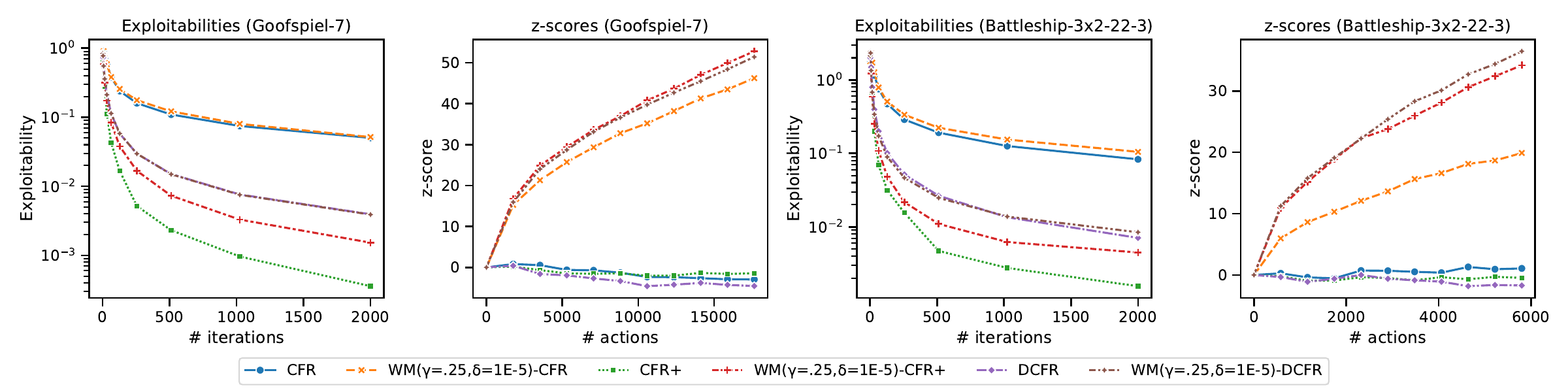}
	\caption{
		Exploitability and z-score plots from our experiments running the (watermarked) CFR family of algorithms.
	}
	\label{fig:experiments}
\end{figure*}

\begin{table*}[t]
	\centering
	\begin{tabular}{l||ccccc|ccccc}
	\toprule
		\multirow{2}{*}{Variant} & \multicolumn{5}{c|}{Goofspiel-7} & \multicolumn{5}{c}{Battleship-3x2-22-3} \\
		& $\epsilon_{t=2000}$ & $n_{z=4}$ & $z_{n=10}$ & $z_{n=100}$ & $z_{n=1000}$ & $\epsilon_{t=2000}$ & $n_{z=4}$ & $z_{n=10}$ & $z_{n=100}$ & $z_{n=1000}$ \\
		\midrule
		CFR & 5.03$\times$10\textsuperscript{$-$2} & N/A & $-$0.68 & $-$1.2 & 0.12 & 8.30$\times$10\textsuperscript{$-$2} & N/A & 0.35 & 0.51 & $-$0.18 \\
		\textbf{WM}-CFR & 5.16$\times$10\textsuperscript{$-$2} & 55 & 2.0 & 4.6 & 12 & 1.05$\times$10\textsuperscript{$-$1} & 277 & $-$0.12 & 1.5 & 8.1 \\
		CFR+ & 3.58$\times$10\textsuperscript{$-$4} & N/A & 1.3 & $-$1.2 & $-$0.039 & 1.56$\times$10\textsuperscript{$-$3} & N/A & 0.81 & $-$0.073 & $-$1.1 \\
		{*}\textbf{WM}-CFR+ & 1.53$\times$10\textsuperscript{$-$3} & 86 & 0.88 & 4.5 & 13 & 4.47$\times$10\textsuperscript{$-$3} & 92 & $-$0.12 & 4.0 & 14 \\
		DCFR & 3.93$\times$10\textsuperscript{$-$3} & N/A & 0.098 & $-$0.43 & 0.90 & 7.07$\times$10\textsuperscript{$-$3} & N/A & 0.35 & 0.80 & $-$1.7 \\
		\textbf{WM}-DCFR & 3.90$\times$10\textsuperscript{$-$3} & 112 & 0.098 & 3.4 & 12 & 8.45$\times$10\textsuperscript{$-$3} & 63 & 0.35 & 4.6 & 15 \\
		\bottomrule
	\end{tabular}
	\caption{
		Tabulated results from our experiments running (watermarked) CFR variants on select games.
		WM-$\mathfrak R$ corresponds to the watermarked counterpart of the regret minimizer $\mathfrak R$ with parameters $\gamma = 0.25$ and $\delta = 10^{-5}$.
		N/A stands for `not applicable'.
		$\epsilon_{t=2000}$ denotes the exploitability of the computed strategy profile after 2,000 iterations, $n_{z=4}$ denotes the number of actions taken for the z-score to cross the $z = 4$ threshold, and finally $z_{n=(\cdot)}$ denotes the z-score after $(\cdot)$ actions have been played.
		The row marked with `*' contains the same data shown in the identically marked rows in Table~\ref{tab:ablations}.
	}
	\label{tab:experiments}
\end{table*}

In our experiments, we tested our watermark on two very different games with tens of millions of nodes each: Goofspiel-7 and Battleship-3x2-22-3. 
The games are described in Appendix~\ref{sec:game-descriptions}.
We considered three foundational learning algorithms from the CFR family: CFR~\cite{zinkevichetal2007}, CFR\textsuperscript{+}~\cite{tammelin2014}, and DCFR~\cite{brownandsandholm2019aaai}.
We also ran their watermarked counterparts with parameters $\gamma = 0.25$ and $\delta = 10^{-5}$, and we further scaled $\delta$ by $\max_{i,j} |A_{i,j}|$, where $A$ is the utility matrix from the perspective of the row player.
We found these to be good choices later when we ablated on the parameter values.

Each (watermarked) CFR variant was run for 2,000 iterations, during which we recorded the average strategy profile's exact exploitability by computing best responses. 
We also ran Monte Carlo game simulations where we pit the resulting watermarked agents against their original counterparts.
To ensure the diversity of play, we initiated the gameplays from every possible early-game situation (first five player moves) and player position permutations.
Our practice is consistent with how chess engines are evaluated~\cite{kimetal}.
We tracked the z-scores across the simulations.

The exploitabilities and z-scores are plotted in Figure~\ref{fig:experiments} and tabulated in Table~\ref{tab:experiments}. The expected utilities of the watermarked solutions against their non-watermarked counterparts are tabulated in Appendix~\ref{sec:omitted-data}.
For CFR and DCFR, the presence of the watermark barely affects the convergence rate, although it does affect CFR\textsuperscript{+} to a small degree.
It is interesting how the different CFR variants react differently to perturbations, and it seems that the way CFR\textsuperscript{+} uniquely behaves (\textit{e.g.}, aggressively flooring negative regrets) may be responsible for the slowdown in convergence.
%That said, this can also be interpreted as a positive sign that the learning algorithm is `embedding' the desired watermark.

As expected, the z-scores of the watermarked agents rise as the number of simulations increases, whereas those of the non-watermarked agents hover near zero as desired. 
We observed that it takes around 100 actions for the z-score corresponding to the advanced CFR variants (CFR\textsuperscript{+} and DCFR) to cross the suggested $z=4$ threshold.
Assuming it takes around 30 seconds (\textit{e.g.}, for a human) to act, this means a watermarked agent would be detected after a couple of hours of gameplay.
We also observed that watermarking does not affect the output solution's expected utility against its non-watermarked counterpart.
Appendix~\ref{sec:roc} contains ROC curves per 100 actions, whose AUCs are above 0.9 for the advanced CFR variants, indicating strong discrimination.

\section{Ablation Studies}

We ablated on the watermark parameters to see how they affect the behavior of the watermarked learning algorithms.
First, we fixed $\delta = 10^{-5}$, iterated through the values of $\gamma \in (0.1, 0.25, 0.5, 0.75, 0.9)$.
Next, we fixed $\gamma = 0.25$, iterated through the values of $\delta \in (0.1, 0.01, 10^{-3}, 10^{-5}, 10^{-8})$.
In our ablation studies, we only considered CFR\textsuperscript{+}.
We also recorded exploitabilities and ran Monte Carlo simulations as we did in the main experiments in the previous section.

The exploitabilities and z-scores are plotted in Figure~\ref{fig:ablations} and tabulated in Table~\ref{tab:ablations}. The expected utilities of the watermarked solutions against the solution from CFR\textsuperscript{+} are tabulated in Appendix~\ref{sec:omitted-data}.
Interestingly, increasing $\gamma$ helped CFR\textsuperscript{+} converge faster, perhaps because this made the perturbations more uniform.
Also, middling gamma values helped the watermarked agent be more detectable.
As expected, decreasing $\delta$ helped CFR\textsuperscript{+} converge faster, as it decreased the magnitudes of the applied perturbations.
That said, increasing $\delta$ helped the watermarked agent be more detectable.
We also observed that, for varying $\gamma$, watermarking does not affect the output solution's head-to-head expected utility against that of CFR\textsuperscript{+}; however, larger values of $\delta$ do so negatively.

\section{Discussion}

We now discuss some potential limitations and possibilities.

\subsection{Attacking the Watermark}

\begin{figure*}[t!]
	\centering
	\includegraphics[width=\textwidth]{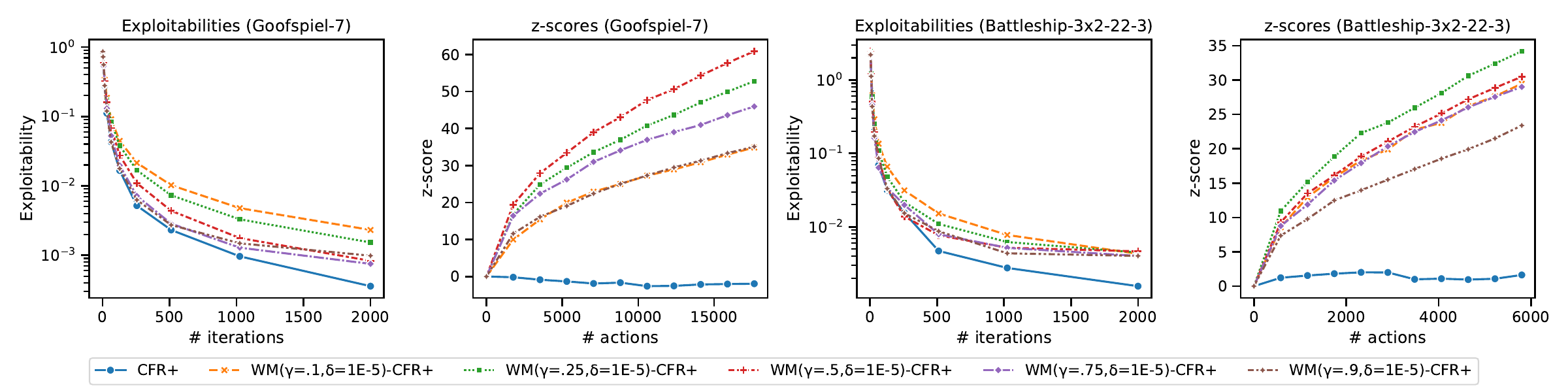}
	\includegraphics[width=\textwidth]{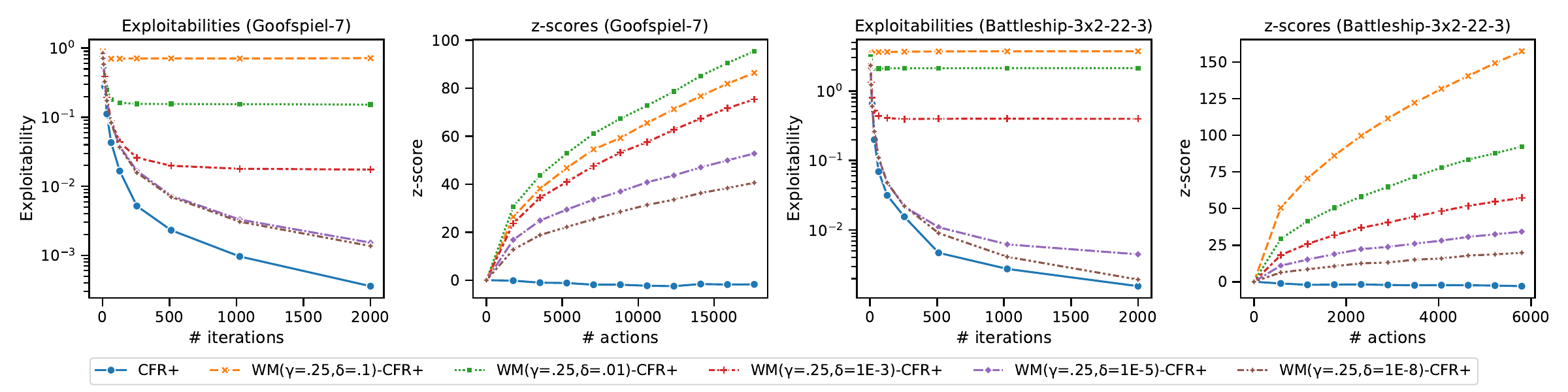}
	\caption{
		Exploitability and z-score plots from our ablation studies running (watermarked) CFR\textsuperscript{+}.
	}
	\label{fig:ablations}
\end{figure*}

\begin{table*}[t!]
	\centering
	\begin{tabular}{c|c||ccccc|ccccc}
	\toprule
		\multirow{2}{*}{$\gamma$} & \multirow{2}{*}{$\delta$} & \multicolumn{5}{c|}{Goofspiel-7} & \multicolumn{5}{c}{Battleship-3x2-22-3} \\
		& & $\epsilon_{t=2000}$ & $n_{z=4}$ & $z_{n=10}$ & $z_{n=100}$ & $z_{n=1000}$ & $\epsilon_{t=2000}$ & $n_{z=4}$ & $z_{n=10}$ & $z_{n=100}$ & $z_{n=1000}$ \\
		\midrule
		0.1 & \multirow{5}{*}{10\textsuperscript{$-$5}} & 2.32$\times$10\textsuperscript{$-$3} & 163 & $-$0.85 & 3.4 & 8.2 & 4.30$\times$10\textsuperscript{$-$3} & 155 & $-$0.33 & 1.7 & 12 \\
		{*}0.25 & & 1.53$\times$10\textsuperscript{$-$3} & 86 & 0.88 & 4.5 & 13 & 4.47$\times$10\textsuperscript{$-$3} & 92 & $-$0.12 & 4.0 & 14 \\
		0.5 & & 8.30$\times$10\textsuperscript{$-$4} & 101 & 0.85 & 3.8 & 15 & 4.64$\times$10\textsuperscript{$-$3} & 65 & 2.2 & 4.0 & 12 \\
		0.75 & & 7.51$\times$10\textsuperscript{$-$4} & 121 & 0.29 & 3.1 & 13 & 4.01$\times$10\textsuperscript{$-$3} & 83 & 1.5 & 4.3 & 11 \\
		0.9 & & 9.85$\times$10\textsuperscript{$-$4} & 203 & $-$0.28 & 1.6 & 8.5 & 4.02$\times$10\textsuperscript{$-$3} & 137 & 1.0 & 3.6 & 8.7 \\
		\midrule
		\multirow{5}{*}{0.25} & 0.1 & 7.16$\times$10\textsuperscript{$-$1} & 29 & 1.7 & 6.4 & 20 & 3.72 & 3 & 5.0 & 21 & 66 \\
		& 0.01 & 1.52$\times$10\textsuperscript{$-$1} & 15 & 2.4 & 7.6 & 24 & 2.13 & 15 & 2.7 & 11 & 38 \\
		& 10\textsuperscript{$-$3} & 1.75$\times$10\textsuperscript{$-$2} & 70 & 0.49 & 6.1 & 20 & 3.97$\times$10\textsuperscript{$-$1} & 16 & 2.7 & 6.6 & 24 \\
		& {*}10\textsuperscript{$-$5} & 1.53$\times$10\textsuperscript{$-$3} & 86 & 0.88 & 4.5 & 13 & 4.47$\times$10\textsuperscript{$-$3} & 92 & $-$0.12 & 4.0 & 14 \\
		& 10\textsuperscript{$-$8} & 1.37$\times$10\textsuperscript{$-$3} & 112 & 0.88 & 3.4 & 11 & 1.94$\times$10\textsuperscript{$-$3} & 136 & $-$0.12 & 2.6 & 7.4 \\
		\bottomrule
	\end{tabular}
	\caption{
		Tabulated results from our ablation studies running watermarked CFR+ on select games.
		N/A stands for `not applicable'.
		$\epsilon_{t=2000}$ denotes the exploitability of the computed strategy profile after 2,000 iterations, $n_{z=4}$ denotes the number of actions taken for the z-score to cross the $z = 4$ threshold, and finally $z_{n=(\cdot)}$ denotes the z-score after $(\cdot)$ actions have been played.
		The rows marked with `*' contain the same data shown in the identically marked row in Table~\ref{tab:experiments}.
	}
	\label{tab:ablations}
\end{table*}

We begin by discussing possible attacks on our watermark. As observed by~\citet{kimetal}, there is no point in modifying the gameplay post-hoc, as utilities have already been realized and other players can also keep records.
Thus, the unique nature of game playing prevents attackers from carrying out `offline' attacks for LLM watermarks involving paraphrasing, insertion, and replacement.
As usual, tricks such as only using the watermarked agent during crucial situations will not suffice in bypassing the watermark, as the z-score will eventually cross the threshold.
However, it is unknown whether the attacker could somehow use peripheral information---such as the expected utilities---to bypass the watermark. If this turns out to be possible, it would imply a major breakthrough in equilibrium finding.

A notable vulnerability of the KFS watermark is less of an issue in imperfect-information games. 
In the KFS watermark, an attacker who is trying to circumvent the watermark can, instead of querying the strategy at the current history $h$, query another history $h'$ whose observation is different from that of $h$ but which is strategically similar.
The attacker would successfully bypass the KFS watermark if the recommended action for $h'$ is also good for $h$.
While, in theory, this attack can be applied to our watermark (\textit{e.g.}, by passing a similar infoset), this is less of a concern in our setting.
Solving most games in real-life requires using game abstraction~\cite{sandholm}, which shrinks the game size by merging strategically similar infosets.
This effectively nullifies the attack since using a good abstraction prevents the attacker from finding a strategically similar infoset to manipulate.
So, interestingly, abstraction techniques help not only to solve larger games more quickly but also to strengthen our watermark!

\subsection{Lossless Watermarking?}

We continue our discussions by raising the possibility of lossless watermarking, that is, ensuring that perturbed regret minimization incurs sublinear regret.
One way to do so is to ensure that perturbations decay exponentially during game solving (see Proposition~\ref{ppn:lossless-probability-simplex}).
However, it is unknown whether a `lossless' P/RM truly `selects' an equilibrium which plays \textcolor{Green}{\textbf{green}} actions more frequently than others.
This connects to the open question raised by~\citet{brownandsandholm2017} on which particular equilibrium CFR converges to.\footnote{See the footnote in~\citet[page 290]{brownandsandholm2017}.}
Additionally, a game's equilibrium space may not even be rich enough to contain a lossless watermarked strategy that can be efficiently detected. As an extreme case, consider a game that has just one equilibrium. 
More details are given in Appendix~\ref{sec:lossless-perturbed-regret-minimizer}.

\section{Conclusions and Future Research}

We developed the first watermarking technique for game-playing agents that also applies to imperfect-information games. Our approach does  watermarking and game solving simultaneously via perturbed regret minimization, which applies perturbations to the observed utilities so as to encourage the learning algorithm to embed the watermark. 
While our watermark incurs some cost in exploitability, we analyzed and bounded this loss.
Empirically, we showed that the costs in exploitability, convergence rate, and expected utility are small to none.
In our ablation studies, we swept through different values of watermark parameters and noted how the learning process differs.
Then, we discussed the possible attacks to our watermark and the potential ways our watermark can be improved.
A promising avenue of future work is expanding our approach to achieve theoretically lossless watermarking in games where that may be possible, and our proposal in the discussion may be a good starting point.

\bibliography{aaai2027}

% If so, you can uncomment the following line and adjust the path to include it.

\clearpage
\appendix

\setcounter{secnumdepth}{2}

\section{Omitted Proofs}
\label{sec:omitted-proofs}

This section contains the proofs omitted from the main body of the paper.

\subsection{Proof of Proposition~\ref{ppn:probability-simplex}}

\begin{proof}
	Let $\vec w$ be the static perturbation.
	Then, we have that
	\begin{align*}
		W^{(T)}
			& = \max_{\vec x \in \Delta^m} \sum_{t = 1}^T \langle\vec u^{(t)}, \vec x\rangle - \sum_{t = 1}^T \langle\vec u^{(t)}, \vec x^{(t)}\rangle \\
			& = \max_{\vec x \in \Delta^m} \sum_{t = 1}^T \langle\vec u^{(t)} + \vec w - \vec w, \vec x\rangle \\
			& \qquad - \sum_{t = 1}^T \langle\vec u^{(t)} + \vec w - \vec w, \vec x^{(t)}\rangle \\
			& \le \max_{\vec x \in \Delta^m} \sum_{t = 1}^T \langle\vec u^{(t)} + \vec w, \vec x\rangle + \max_{\vec x \in \Delta^m} \sum_{t = 1}^T \langle-\vec w, \vec x\rangle \\
			& \qquad - \sum_{t = 1}^T \langle\vec u^{(t)} + \vec w, \vec x^{(t)}\rangle - \sum_{t = 1}^T \langle-\vec w, \vec x^{(t)}\rangle \\
			& = R^{(T)} + \max_{\vec x \in \Delta^m} \sum_{t = 1}^T \langle-\vec w, \vec x\rangle - \sum_{t = 1}^T \langle-\vec w, \vec x^{(t)}\rangle \\
			& \le R^{(T)} - \sum_{t = 1}^T \min_{\vec x \in \Delta^m} \langle\vec w, \vec x\rangle + \sum_{t = 1}^T \langle\vec w, \vec x^{(t)}\rangle \\
			& \le R^{(T)} - TL + TU \\
			& = R^{(T)} + T(U - L),
	\end{align*}
	as required.
\end{proof}

\subsection{Proof of Proposition~\ref{ppn:sequence-form-polytope}}

\begin{proof}
	Let $\mathfrak W_j$ be the local P/RM in the decision point $j$, and let $W_j^{(T)}$ be the regret incurred by $\mathfrak W_j$ until time $T$.
	\citet[Proposition 4.1]{farina} showed that
	\[
		W^{(T)} \le \sum_{j \in \mathcal J} \max\{0, W_j^{(T)}\}.
	\]
	By Proposition~\ref{ppn:probability-simplex}, we have that
	\begin{align*}
		W^{(T)}
			& \le \sum_{j \in \mathcal J} \max\{0, R_j^{(T)} + T(U - L)\} \\
			& \le \sum_{j \in \mathcal J} \max\{0, R_j^{(T)}\} + T|\mathcal J|(U - L),
	\end{align*}
	as required.
\end{proof}

\subsection{Proof of Theorem~\ref{thm:eps-ne}}

\begin{proof}
	Let $\mathcal X$ and $\mathcal Y$ be the sequence-form polytopes of the row and column players, respectively; let $A$ be the utility matrix from the perspective of the row player, let $(\bar x, \bar y)$ be the average sequence-form strategy profile, and let $\epsilon$ be its exploitability.
	Further, let $(\vec x^{(t)}, \vec y^{(t)})$ be the output strategy profile at time $t$.
	Then, we have that
	\begin{align*}
		2\epsilon
			& = \max_{\vec x \in \mathcal X} \vec x^\top A \bar y - \min_{\vec y \in \mathcal Y} \bar x^\top A \vec y \\
			& = \max_{\vec x \in \mathcal X} \langle A \bar y, \vec x\rangle + \max_{\vec y \in \mathcal Y} \langle-A \bar x, \vec y\rangle \\
			& = \max_{\vec x \in \mathcal X} \sum_{t = 1}^T \langle A \vec y^{(t)}, \vec x\rangle/T + \max_{\vec y \in \mathcal Y} \sum_{t = 1}^T \langle-A \vec x^{(t)}, \vec y\rangle/T \\
			& = \left(\max_{\vec x \in \mathcal X} \sum_{t = 1}^T \langle A \vec y^{(t)}, \vec x\rangle - \sum_{t = 1}^T \langle A \vec y^{(t)}, \vec x^{(t)}\rangle\right)/T \\
			& \qquad + \left(\max_{\vec y \in \mathcal Y} \sum_{t = 1}^T \langle-A \vec x^{(t)}, \vec y\rangle - \sum_{t = 1}^T \langle -A \vec x^{(t)}, \vec y^{(t)}\rangle\right)/T.
	\end{align*}
	Let $W_\mathcal X^{(T)}$ and $W_\mathcal Y^{(T)}$ be the regret incurred by the regret minimizers operating over the row and column sequence-form polytopes until time $T$, and let $\mathcal J_\mathcal X$ and $\mathcal J_\mathcal Y$ be the decision points in the row and column TFSDPs.
	By Proposition~\ref{ppn:sequence-form-polytope}, we have that
	\begin{align*}
		2\epsilon
			& = W_\mathcal X^{(T)}/T + W_\mathcal Y^{(T)}/T \\
			& \le \left(\sum_{j \in \mathcal J_\mathcal X} \max\{0, R_j^{(T)}\} + T|\mathcal J_\mathcal X|(U - L)\right)/T \\
			& \qquad + \left(\sum_{j \in \mathcal J_\mathcal Y} \max\{0, R_j^{(T)}\} + T|\mathcal J_\mathcal Y|(U - L)\right)/T \\
			& = \sum_{j \in \mathcal J_\mathcal X \cup \mathcal J_\mathcal Y} \max\{0, R_j^{(T)}\}/T + |\mathcal J_\mathcal X \cup \mathcal J_\mathcal Y|(U - L),
	\end{align*}
	where $L$ and $U$ are the minimum and maximum perturbations.
	Since decision points and infosets are equivalent and $L = -\delta\gamma/(1 - \gamma), U = \delta$, we have that
	\[
		2\epsilon \le \sum_{I \in \mathcal I} \max\{0, R_I^{(T)}\}/T + \delta|\mathcal I|/(1 - \gamma),
	\]
	as required.
\end{proof}

\subsection{Proof of Corollary~\ref{cll:cfr}}

\begin{proof}
	Let $R_I^{(T)}$ be the regret incurred by the component of the P/RM at the decision point corresponding to infoset $I$ until time $T$.
	By Theorem~\ref{thm:eps-ne}, we have that the average strategy profile is a $2\epsilon$-NE, where
	\[
		2\epsilon \le \sum_{I \in \mathcal I} \max\{0, R_I^{(T)}\}/T + \delta|\mathcal I|/(1 - \gamma).
	\]
	The type of local regret minimizers used in CFR is regret matching, which satisfies $R_I^{(T)} \le \Delta'\sqrt{|A(I)|}\sqrt{T}$, where $\Delta'$ is the range of utilities observed by the local regret minimizer~\cite{cesabianchiandlugosi}.
	Since the component of a P/RM observes perturbed utilities, we have that $\Delta' \le \Delta + \delta/(1 - \gamma)$.
	Also, since $|A(I)| \le m$ for any $I \in \mathcal I$, we have that the average strategy profile is a $2\epsilon$-NE where
	\[
		2\epsilon \le |\mathcal I|(\Delta + \delta/(1 - \gamma))\sqrt m/\sqrt T + \delta|\mathcal I|/(1 - \gamma),
	\]
	as required.
\end{proof}

\section{Game Descriptions}
\label{sec:game-descriptions}

We tested several CFR variants and their watermarked counterparts using two games: Goofspiel-7 and Battleship-3x2-22-3.
These are common benchmark games that are codified in OpenSpiel 1.6.15~\cite{lanctot}, a popular computational game-theory library.
Goofspiel-7 is the imperfect-information variant of the game of Goofspiel, played with seven cards and descending point order.
Battleship-3x2-22-3 is simply the game of battleship~\cite{farinaetal2019b} played on a 3-by-2 board with two ships of size 2 and value 4, where every player fires three shots.
Goofspiel-7 has about 4.8$\times$10\textsuperscript{7} nodes in its game tree, whereas battleship-3x2-22-3 has around 2.7$\times$10\textsuperscript{7} nodes.

\section{Omitted Data from the Experiments}
\label{sec:omitted-data}

\begin{table}[t]
	\centering
	\begin{tabular}{l|cc}
	\toprule
		Variant & Goofspiel-7 & Battleship-3x2-22-3 \\
		\midrule
		\textbf{WM}-CFR & \cellcolor{RedOrange} $-$7.93$\times$10\textsuperscript{$-$4} & \cellcolor{green} 7.49$\times$10\textsuperscript{$-$4} \\
		{*}\textbf{WM}-CFR+ & \cellcolor{green} 4.34$\times$10\textsuperscript{$-$5} & \cellcolor{green} 5.85$\times$10\textsuperscript{$-$4} \\
		\textbf{WM}-DCFR & \cellcolor{green} 1.92$\times$10\textsuperscript{$-$4} & \cellcolor{RedOrange} $-$7.40$\times$10\textsuperscript{$-$4} \\
		\bottomrule
	\end{tabular}
	\caption{
		Expected utilities of the solutions obtained from watermarked CFR variants against their non-watermarked counterparts.
		The row marked with `*' contains the same data shown in the identically marked rows in Table~\ref{tab:ablations2}.
	}
	\label{tab:experiments2}
\end{table}

\begin{table}[t!]
	\centering
	\begin{tabular}{c|c||cc}
	\toprule
		$\gamma$ & $\delta$ & Goofspiel-7 & Battleship-3x2-22-3 \\
		\midrule
		0.1 & \multirow{5}{*}{10\textsuperscript{$-$5}} & \cellcolor{green} 6.10$\times$10\textsuperscript{$-$5} & \cellcolor{green} 1.27$\times$10\textsuperscript{$-$4} \\
		{*}0.25 & & \cellcolor{green} 4.34$\times$10\textsuperscript{$-$5} & \cellcolor{green} 5.85$\times$10\textsuperscript{$-$4} \\
		0.5 & & \cellcolor{green} 1.26$\times$10\textsuperscript{$-$5} & \cellcolor{green} 5.55$\times$10\textsuperscript{$-$4} \\
		0.75 & & \cellcolor{green} 5.83$\times$10\textsuperscript{$-$6} & \cellcolor{green} 1.72$\times$10\textsuperscript{$-$4} \\
		0.9 & & \cellcolor{green} 4.62$\times$10\textsuperscript{$-$6} & \cellcolor{green} 2.60$\times$10\textsuperscript{$-$5} \\
		\midrule
		\multirow{5}{*}{0.25} & 0.1 & \cellcolor{green} 4.75$\times$10\textsuperscript{$-$3} & \cellcolor{RedOrange} $-$8.88$\times$10\textsuperscript{$-$2} \\
		& 0.01 & \cellcolor{RedOrange} $-$2.95$\times$10\textsuperscript{$-$5} & \cellcolor{RedOrange} $-$1.35$\times$10\textsuperscript{$-$1} \\
		& 10\textsuperscript{$-$3} & \cellcolor{green} 1.11$\times$10\textsuperscript{$-$4} & \cellcolor{green} 3.74$\times$10\textsuperscript{$-$2} \\
		& {*}10\textsuperscript{$-$5} & \cellcolor{green} 4.34$\times$10\textsuperscript{$-$5} & \cellcolor{green} 5.85$\times$10\textsuperscript{$-$4} \\
		& 10\textsuperscript{$-$8} & \cellcolor{green} 4.46$\times$10\textsuperscript{$-$5} & \cellcolor{RedOrange} $-$3.62$\times$10\textsuperscript{$-$5} \\
		\bottomrule
	\end{tabular}
	\caption{
		Expected utilities of the solutions obtained from watermarked CFR+ against that from CFR+.
		The rows marked with `*' contain the same data shown in the identically marked row in Table~\ref{tab:experiments2}.
	}
	\label{tab:ablations2}
\end{table}

The expected utilities of the watermarked solutions from the experiments against their non-watermarked counterparts are tabulated in Table~\ref{tab:experiments2}, and the expected utilities of the watermarked solutions from our ablation studies against that from CFR\textsuperscript{+} are tabulated in Table~\ref{tab:experiments2}.
In general, watermarking does not seem to have much of an effect on the expected utilities of the solutions against those produced from the non-watermarked CFR variants, except when $\delta$ is large enough.

\subsection{ROC Curves}
\label{sec:roc}

\begin{figure}[t!]
	\centering
	\includegraphics[width=\columnwidth]{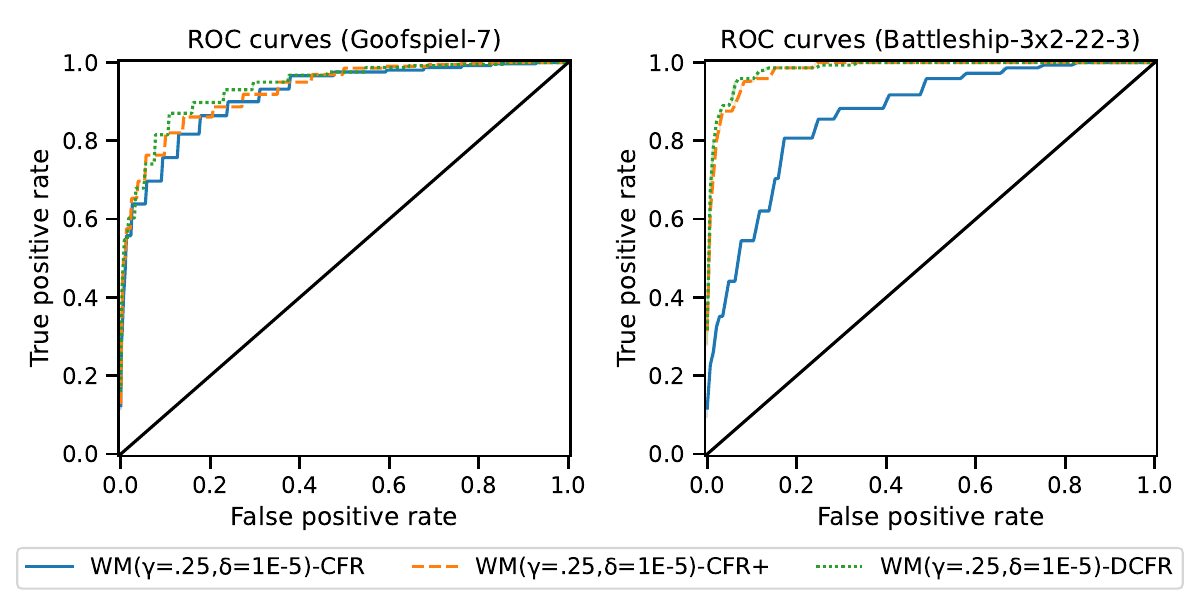}
	\includegraphics[width=\columnwidth]{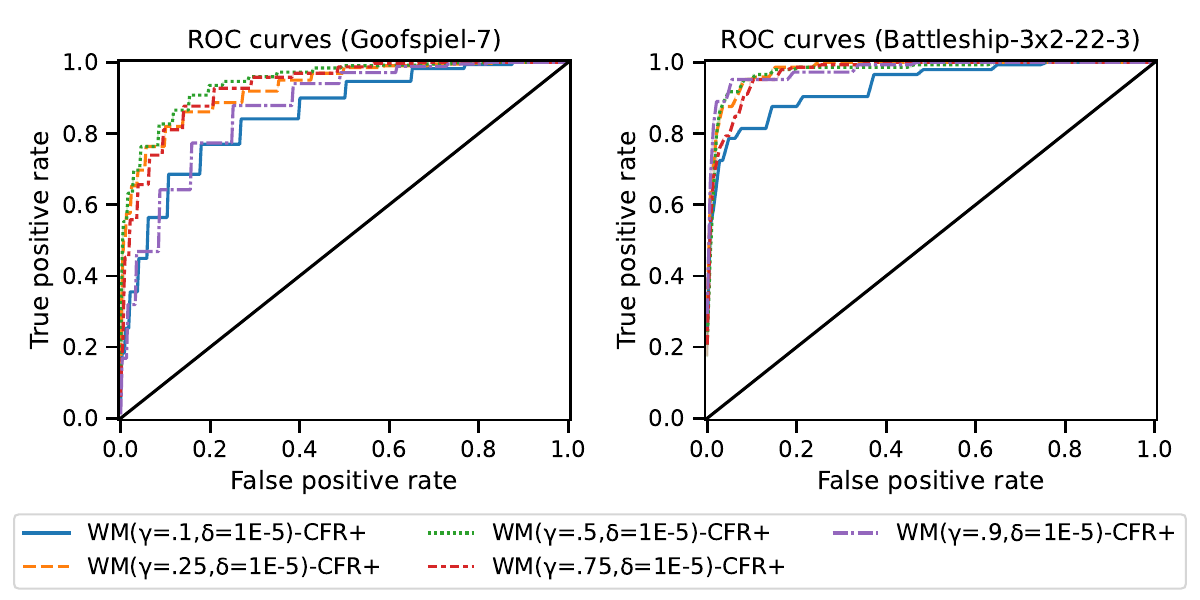}
	\includegraphics[width=\columnwidth]{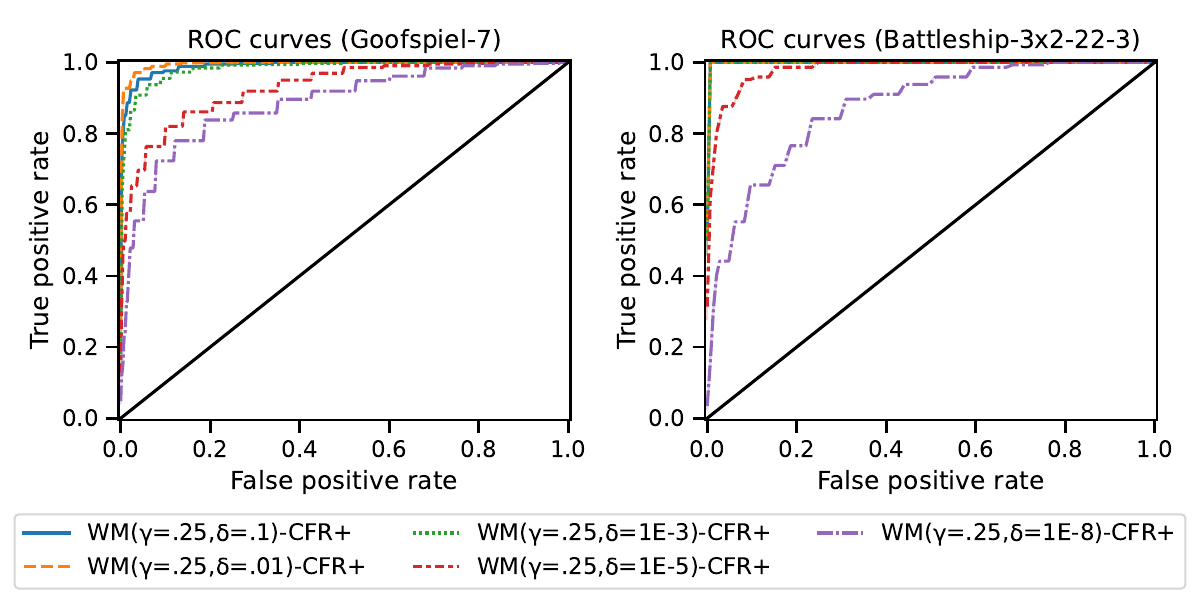}
	\caption{
		Rate of characteristic curves for the watermarked agents in the experiments and ablation studies.
	}
	\label{fig:roc}
\end{figure}

\begin{table}[t]
	\centering
	\begin{tabular}{l|cc}
	\toprule
		Variant & Goofspiel-7 & Battleship-3x2-22-3 \\
		\midrule
		\textbf{WM}-CFR & 0.917 & 0.868 \\
		{*}\textbf{WM}-CFR+ & 0.929 & 0.982 \\
		\textbf{WM}-DCFR & 0.937 & 0.985 \\
		\bottomrule
	\end{tabular}
	\caption{
		AUC values of the ROC curves for the watermarked CFR variants in our main experiments (not ablations).
		The row marked with `*' contains the same data shown in the identically marked rows in Table~\ref{tab:roc2}.
	}
	\label{tab:roc}
\end{table}

\begin{table}[t]
	\centering
	\begin{tabular}{c|c||cc}
	\toprule
		$\gamma$ & $\delta$ & Goofspiel-7 & Battleship-3x2-22-3 \\
		\midrule
		0.1 & \multirow{5}{*}{10\textsuperscript{$-$5}} & 0.855 & 0.935 \\
		{*}0.25 & & 0.929 & 0.982 \\
		0.5 & & 0.946 & 0.977 \\
		0.75 & & 0.931 & 0.974 \\
		0.9 & & 0.872 & 0.980 \\
		\midrule
		\multirow{5}{*}{0.25} & 0.1 & 0.991 & 1.000 \\
		& 0.01 & 0.995 & 1.000 \\
		& 10\textsuperscript{$-$3} & 0.983 & 1.000 \\
		& {*}10\textsuperscript{$-$5} & 0.929 & 0.982 \\
		& 10\textsuperscript{$-$8} & 0.884 & 0.875 \\
		\bottomrule
	\end{tabular}
	\caption{
		AUC values of the ROC curves for watermarked CFR\textsuperscript{+} in our ablation studies.
		The row marked with `*' contains the same data shown in the identically marked rows in Table~\ref{tab:roc}.
	}
	\label{tab:roc2}
\end{table}

From the play data of each watermarked engine, we batched every non-overlapping sequence of 100 actions.
We chose the number 100, as we deemed this was enough to detect the watermark when more practical CFR variants (\textit{e.g.}, CFR\textsuperscript{+} and DCFR) are used.
Then, we swept through z-thresholds and plotted the rate of characteristic (ROC) curves in Figure~\ref{fig:roc}.
Their area under curve (AUC) values are tabulated in Tables~\ref{tab:roc} and~\ref{tab:roc2}.
For more practical CFR variants, we see that the AUC values are above 0.9, indicating strong discrimination.
In our ablations, middling values of $\gamma$ and higher values of $\delta$ lead to higher AUC values.

\section{Lossless P/RM}
\label{sec:lossless-perturbed-regret-minimizer}

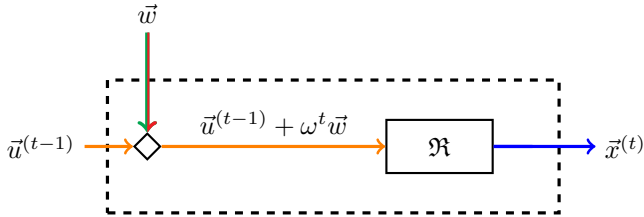
\begin{figure}[t!]
	\centering
	\begin{tikzpicture}
		\node[draw, dashed, minimum width=17em, minimum height=5em, very thick] (boundary) at (0em, 0em) {};
		\node (utility) at (-11em, 0) {$\vec u^{(t - 1)}$};
		\node (delta) at (-7em, 5em) {$\vec w$};
		\node[diamond, draw, minimum size=1em, thick, inner sep=0em] (gate) at (-7em, 0em) {};
		\node[draw, minimum width=4em, minimum height=2em, thick] (regret minimizer) at (4em, 0em) {$\mathfrak R$};
		\node (decision) at (11em, 0em) {$\vec x^{(t)}$};

		\draw (delta) edge[->, color=Green, ultra thick] (gate);
		\begin{scope}
			\clip (-7em, 5em) -- (-7em, 0em) -- (-5em, 5em) -- (-5em, 0em) -- cycle;
			\draw (delta) edge[->, color=Red, ultra thick] (gate);
		\end{scope}
		\draw (utility) edge[->, color=orange, very thick] (gate);
		\draw (gate) edge[->, color=orange, very thick] node[above, color=black] {$\vec u^{(t - 1)} + \omega^t \vec w$} (regret minimizer);
		\draw (regret minimizer) edge[color=blue, very thick] (7.5em, 0em);
		\draw (regret minimizer) edge[->, color=blue, very thick] (decision);
	\end{tikzpicture}
	\caption{
		Regret circuit diagram of a lossless P/RM with $\mathfrak R$ as its component and perturbation $\vec w$.
	}
	\label{fig:lossless-perturbed-regret-minimizer}
\end{figure}

\begin{algorithm}[t!]
	\caption{Lossless regret minimizer under perturbation (P/RM)}
	\label{alg:lossless-perturbed-regret-minimizer}
	\textbf{Input:} A component regret minimizer $\mathfrak R$ operating over $\Delta^m$, a perturbation $\vec w$, and a discount rate $\omega \in (0, 1)$.
	\BlankLine
	\textsc{NextStrategy}$()$: \\
	\Indp{
		\Return $\mathfrak R$.\textsc{NextStrategy}$()$\;
	}
	\Indm
	\BlankLine
	\textsc{ObserveUtility}$\left(\vec u^{(t)} \in \mathbb R^m\right)$: \\
	\Indp{
		$\mathfrak R$.\textsc{ObserveUtility}$(\vec u^{(t)} + \omega^t \vec w)$\;
	}
	\Indm
\end{algorithm}

In this section, we formally define a \textit{lossless regret minimizer under perturbation (P/RM)}, which is identical to a P/RM except that it accepts an extra discount rate parameter $\omega \in (0, 1)$ and the perturbation is no longer static.
Its regret circuit diagram is given in Figure~\ref{fig:lossless-perturbed-regret-minimizer}, and its pseudocode is given in Algorithm~\ref{alg:lossless-perturbed-regret-minimizer}.
Just like its lossy counterpart, at time $t \in \mathbb N^+$, the lossless P/RM forwards the output of the component regret minimizer as its own.
However, when the lossless P/RM observes a utility vector $\vec u^{(t)}$, it adds a \textit{discounted} perturbation $\omega^t \vec w$ to it.
Finally, the perturbed utility $\vec u^{(t)} + \omega^t \vec w$ is forwarded to the component regret minimizer for its observation.

In the following proposition, we show that a lossless P/RM incurs sublinear average regret.

\begin{proposition}
	Let $\mathfrak W$ be a lossless P/RM with a component regret minimizer $\mathfrak R$ operating over $\Delta^m$, a perturbation in $[L, U]^m$, and a discount rate $\omega \in (0, 1)$.
	Further, let $W^{(T)}$ and $R^{(T)}$ be the regret incurred by $\mathfrak W$ and $\mathfrak R$, respectively, until time $T$.
	Then,
	\[
		W^{(T)} \le R^{(T)} + \omega (U - L)/(1 - \omega).
	\]
	\label{ppn:lossless-probability-simplex}
\end{proposition}
\begin{proof}
	Let $\vec w$ be the perturbation.
	Then, we have that
	\begin{align*}
		W^{(T)}
			& = \max_{\vec x \in \Delta^m} \sum_{t = 1}^T \langle\vec u^{(t)}, \vec x\rangle - \sum_{t = 1}^T \langle\vec u^{(t)}, \vec x^{(t)}\rangle \\
			& = \max_{\vec x \in \Delta^m} \sum_{t = 1}^T \langle\vec u^{(t)} + \omega^t \vec w - \omega^t \vec w, \vec x\rangle \\
			& \qquad - \sum_{t = 1}^T \langle\vec u^{(t)} + \omega^t \vec w - \omega^t \vec w, \vec x^{(t)}\rangle \\
			& \le \max_{\vec x \in \Delta^m} \sum_{t = 1}^T \langle\vec u^{(t)} + \omega^t \vec w, \vec x\rangle + \max_{\vec x \in \Delta^m} \sum_{t = 1}^T \langle-\omega^t \vec w, \vec x\rangle \\
			& \qquad - \sum_{t = 1}^T \langle\vec u^{(t)} + \omega^t \vec w, \vec x^{(t)}\rangle - \sum_{t = 1}^T \langle-\omega^t \vec w, \vec x^{(t)}\rangle \\
			& = R^{(T)} + \max_{\vec x \in \Delta^m} \sum_{t = 1}^T \langle-\omega^t \vec w, \vec x\rangle - \sum_{t = 1}^T \langle-\omega^t \vec w, \vec x^{(t)}\rangle \\
			& \le R^{(T)} - \sum_{t = 1}^T \omega^t \min_{\vec x \in \Delta^m} \langle\vec w, \vec x\rangle + \sum_{t = 1}^T \omega^t \langle\vec w, \vec x^{(t)}\rangle \\
			& \le R^{(T)} - \sum_{t = 1}^T \omega^t L + \sum_{t = 1}^T \omega^t U \\
			& \le R^{(T)} + \omega (U - L)/(1 - \omega),
	\end{align*}
	as required.
\end{proof}
\begin{remark}
	Assuming $\mathfrak R$ is a regret minimizer, $R^{(T)} = o(T)$, so we have that $W^{(T)} = o(T)$ as well.
	Thus, when lossless P/RMs are used, the resulting perturbed regret minimization algorithm converges to NE when solving 2p0s games.
	However, the existence of a lossless P/RM does not necessarily imply that lossless watermarking is possible.
	Indeed, it is unknown whether this regret minimizer can perform `equilibrium selection' in that the average strategy converges to an NE that plays \textcolor{Green}{\textbf{green}} actions more frequently than others.
	This connects to the open question raised by~\citet{brownandsandholm2017} on which particular equilibrium CFR converges to.\footnote{See the footnote in~\citet[page 290]{brownandsandholm2017}.}
	Also, the equilibrium space may not be large enough to contain a desired `watermarked' strategy profile anyway, although we do know that even toy games of moderate sizes like Kuhn poker~\cite{kuhn} and Leduc poker~\cite{southeyetal} contain a rich space of NE.
\end{remark}

\section{Testbench specification}

Our testbench has an AMD Ryzen 9 3900X 12-core, 24-thread CPU and 128 GB of memory.

\end{document}